\documentclass[conference]{IEEEtran}
\IEEEoverridecommandlockouts
\usepackage{cite}
\usepackage[letterpaper, top=0.7in, bottom=1in, left=0.625in, right=0.625in]{geometry}
\usepackage{amsmath,amssymb,amsfonts}
\usepackage{amsthm} 
\newtheorem{proposition}{Proposition}
\usepackage{booktabs}
\usepackage{algorithmic,algorithm}

\newtheorem{lemma}{Lemma}
\usepackage{caption,subcaption}
\usepackage{graphicx}
\usepackage{csquotes}
\usepackage{textcomp}
\usepackage{xcolor}

\def\BibTeX{{\rm B\kern-.05em{\sc i\kern-.025em b}\kern-.08em
    T\kern-.1667em\lower.7ex\hbox{E}\kern-.125emX}}
\begin{document}

\title{Stay or Switch: Online Conformal Bayesian Optimization Guided Fluid Antenna Configuration 
}

\author{\IEEEauthorblockN{
    Gangyong Zhu,
    Jia Yan,
    and Shijian Gao
}
\IEEEauthorblockA{The Hong Kong University of Science and Technology (Guangzhou), Guangzhou, China}
}

\maketitle
\begin{abstract}
Fluid antenna systems (FAS) introduce additional spatial degrees of freedom to enable integrated sensing and communication (ISAC) in air-ground networks. However, conventional studies often overlook or simplify the physical overheads and switching costs of FAS. In practice, port switching incurs non-negligible time, during which communication and sensing may continue but with potentially degraded slot-level performance. This leads to two key challenges: (1) the characterization of a slot-level, cost-aware ISAC metric is difficult, and (2) the large port space and accompanying abrupt environmental variations demand more reliable online decision-making. To address these challenges, a cost-aware multi-objective FAS switching problem is formulated, jointly considering slot-level ISAC performance and switching energy. The online conformal Bayesian optimization (OCBO) algorithm is then proposed to learn the unknown gray-box ISAC objectives and calibrate surrogate uncertainty for robust stay-or-switch decisions. Simulation results demonstrate that the proposed cost-aware optimization framework achieves substantially improved long-term ISAC performance compared to existing baselines.
\end{abstract}


\section{Introduction}

\IEEEPARstart{F}{uture} 6G networks demand the seamless integration of high-rate communication and high-precision environmental sensing. Fluid antenna systems (FAS), characterized by their ability to move radiation positions within a physical aperture, have emerged as a promising technology for achieving integrated sensing and communication (ISAC) in air-ground systems \cite{GonzalezPrelcic2024ISACRevolution,gao2026iscc_lan,Lu2024ISACRecentAdvances}. By dynamically optimizing antenna positions, FAS can reshape the spatial channel and provide additional spatial degrees of freedom to enhance ISAC integration gain \cite{New2025FASTutorial}. Motivated by this, extensive research efforts have been devoted to port-selection for FAS-ISAC systems \cite{Wu2025FASEnabling6G,Zhang2026FASUNISAC}.

Most existing FAS-ISAC studies optimize the active antenna positions based on the instantaneous performance of each selected configuration, while simplifying the physical switching process as an instantaneous or cost-free switch \cite{Zhu2024MovableAntennaMultiuser,Zou2024FASISACTradeoff}. This idealization may lead to overly greedy port switching, where the system frequently moves toward configurations with higher instantaneous gains but ignores the associated movement delay, and energy consumption. As a result, frequent switching may even compromise the performance gains brought by the spatial flexibility of FAS. To address this issue, recent studies start to consider cost-aware FAS and movable-antenna designs from different perspectives. The author in \cite{liu2026switchingCostFAS} incorporates a switching penalty into the long-term resource-allocation objective, discouraging frequent changes of the active FAS ports and improving the stability of scheduling decisions. Moreover, exhaustive port scanning is shown to become impractical in time-varying channels due to switching latency and measurement offsets in \cite{dinis2026spatiotemporalFAS}. For liquid-metal FAS, movement delay and actuation energy are incorporated into energy-efficiency-oriented port selection, indicating that ignoring these costs can result in overly aggressive movement and degraded performance \cite{zhang2026energyEfficientFAIDET}. The movement-duration tradeoff of the movable-antenna systems further shows that the time spent on antenna movement may improve the subsequent channel condition but reduces the remaining time for data transmission \cite{hu2026fundamentalMATradeoff}. However, these studies mainly model the switching interval as an energy penalty, a measurement delay, or a move-then-transmit overhead. For liquid-metal FAS-ISAC, a more fine-grained question remains: \textit{how should communication and sensing objectives during antenna movement be modeled and exploited for online switching?}

To answer this question, this paper formulates a cost-aware FAS port-switching problem for liquid-metal FAS-ISAC, where the switching process is treated as an unstable but still usable transmission interval. During this interval, communication and sensing signals are not simply discarded; instead, they are accumulated as part of the slot-level cost-aware physical utility together with the dwell-stage performance and mechanical movement energy. Solving this problem faces two key challenges: first, the cost-aware ISAC metric is difficult to characterize, since the slot-level utility jointly depends on switching-stage performance, dwell-stage performance, switching time, and movement energy, while the real-time channel state information (CSI) and target responses required for explicit evaluation are difficult to obtain in dynamic air-ground environments; second, the multi-position FAS configuration space, moving users and targets, and potential abrupt environmental changes may cause model mismatch and overconfident surrogate predictions, requiring more stable online decision-making \cite{yang2025synesthesia}. To address these challenges, we propose an online conformal Bayesian optimization (OCBO) algorithm. Specifically, Bayesian optimization (BO) is employed to learn the unknown mapping from discrete FAS port switching decisions to cost-aware ISAC objectives using limited evaluations, while online conformal calibration adjusts the surrogate uncertainty through a dynamic residual buffer to improve robustness under environmental variations. Finally, the port switching decision is made by maximizing the cost-aware acquisition function that balances the calibrated utility improvement against the switching energy.

\section{System Model}
In this section, we first introduce the network architecture of the considered FAS-ISAC systems, and then establish the communication and sensing system models.

\subsection{Network Architecture}

We consider the FAS-ISAC systems comprising $B$ base stations (BSs), $K$ single-antenna users $\mathcal{K}$, and $L$ targets denoted by $\{ \mathbf{p}_\ell,\alpha_\ell\}_{\ell=1}^{L}$, as illustrated in Fig. \ref{fig:system_model}. Each BS is equipped with a rotatable FAS, where a set of candidate ports $\mathcal{C}_b$ is deployed on a planar aperture. In each time slot, BS $b$ activates $M_b$ transmit ports and $N_b$ receive ports from $\mathcal{C}_b$, ensuring that the activated port sets are disjoint. Specifically, let $\mathcal{M}_{b}^{\text{tx}}$ and $\mathcal{M}_{b}^{\text{rx}}$ denote the local coordinate matrices for the transmit and receive ports, respectively. These coordinates are jointly determined by the selected port indices and the surface orientation defined by rotation angles $(\theta_b, \phi_b)$ \cite{Zheng2025RAComSens}. This geometric configuration is defined as $a_t = \{ \mathcal{M}_{b,t}^{\text{tx}}, \mathcal{M}_{b,t}^{\text{rx}}, \theta_{b,t}, \phi_{b,t} \}_{b\in\mathcal{B}}$, which allows for dynamic alignment of the FAS antenna position at each time slot $t$. Based on the FAS configuration $a_t$, the BSs transmit pilot signals and compute the downlink precoding matrix $\mathbf{W}(a_t)$ for subcarriers $n \in \mathcal{N}$ according to the received feedback, then evaluate the ISAC performance.

The online decision process is structured into discrete time intervals of duration $T_{\mathrm{s}}$, where each slot comprises a switching stage for port switching and a dwell stage for stable communication and sensing. Let $v_{\mathrm{FA}}$ denote the speed of liquid metal. Then, the total switching time $T_{\mathrm{sw}}(a_t | a_{t-1})$ is defined as
\begin{equation}
T_{\mathrm{sw}}(a_t | a_{t-1}) = \max_{b \in \mathcal{B}, m \in \mathcal{M}_b \cup \mathcal{N}_b} \frac{d_{\mathrm{path}}^{(b,m)}(a_{t-1}, a_t)}{v_{\mathrm{FA}}},
\end{equation}
where $d_{\mathrm{path}}^{(b,m)}(a_{t-1}, a_t)$ denotes the distance traveled by the $m$-th port at BS $b$ during the switching interval \cite{Wang2026MAVelocityProfile}.

\subsection{ISAC Model}

The communication and sensing spatial responses are modeled by the field response matrix (FRM) \cite{Chen2025FluidMIMOPortSelection}. For a spatial path with direction $\mathbf{u}(\theta, \phi)$, the transmit steering vector is determined by the selected antenna positions:
\begin{equation}
\mathbf{a}_b(\mathcal{M}_{b,t}^{\text{tx}}, \theta_{b,t}, \phi_{b,t}) = \left[ e^{-j\frac{2\pi}{\lambda}\mathbf{q}_{b,1}^{T}\mathbf{u}}, \dots, e^{-j\frac{2\pi}{\lambda}\mathbf{q}_{b,M_b}^{T}\mathbf{u}} \right]^T,
\end{equation}
where $\left\{\mathbf{q}_{b,m}\right\}^{M_b}_{m=1}$ represents the three-dimensional global positions of the $m$-th transmit port. The receive steering vector $\mathbf{b}_b(\mathcal{M}_{b}^{\mathrm{rx}}, \theta, \phi) \in \mathbb{C}^{N_b^{\mathrm{rx}} \times 1}$ is defined analogously. With $b$, $k$, and $n$ denoting the BS index, user index, and orthogonal frequency division multiplexing (OFDM) subcarrier index, respectively, the communication channel vector $\mathbf h_{b,k,n}$ from BS $b$ to user $k$ over subcarrier $n$ is obtained from the FRM $\mathbf H_{b,k,n}$ as
\begin{equation}
\mathbf h_{b,k,n}
=
\mathbf H_{b,k,n}(\mathcal M_b^{\mathrm{tx}})
\boldsymbol{\sigma}_{b,k,n},
\label{eq:comm_channel_vector}
\end{equation}
where $\mathbf H_{b,k,n}(\mathcal M_b^{\mathrm{tx}})$ denotes the submatrix of $\mathbf H_{b,k,n}$ corresponding to the selected transmit ports.
\begin{figure}[t]
  \centering
  \includegraphics[width=0.95\linewidth]{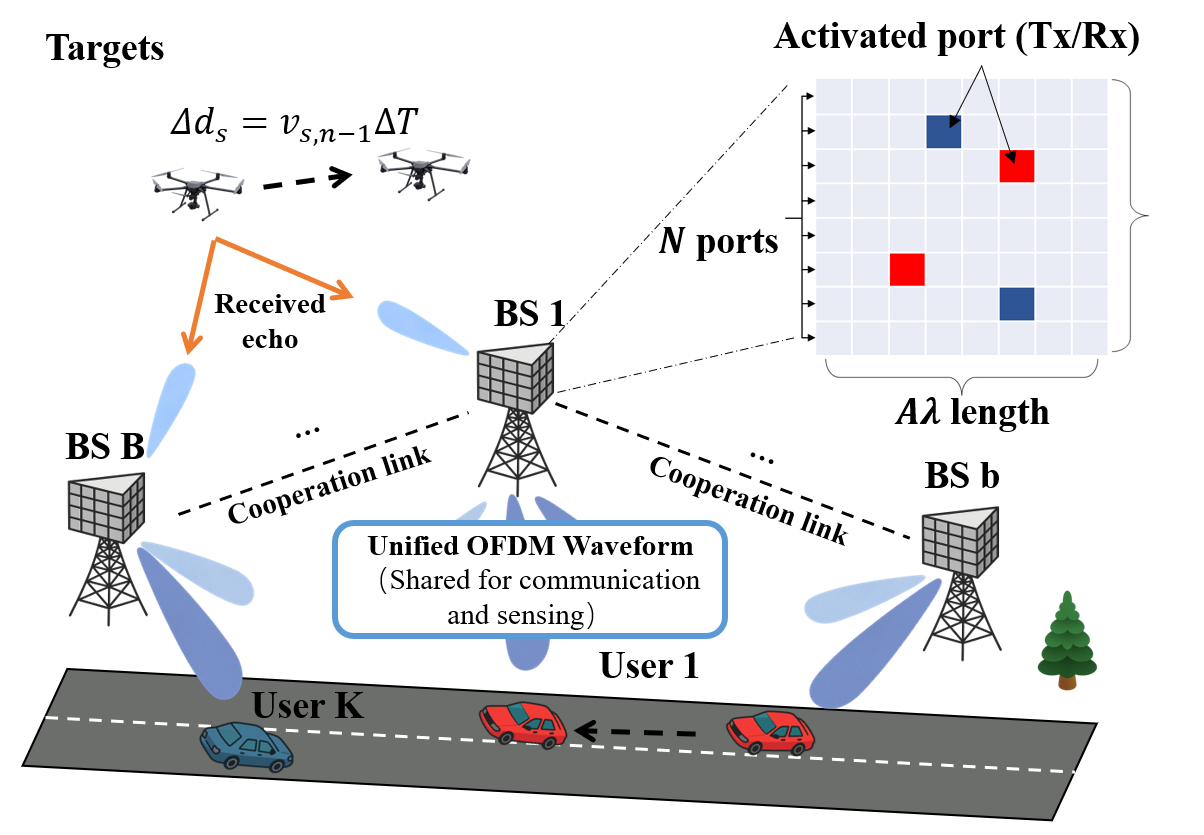}%
  \caption{Structure of the considered FAS-ISAC systems.}
    \label{fig:system_model}
\end{figure}

Adopting the OFDM scheme, the transmitted signal vector $\mathbf{x}_{b,i,n}$ from BS $b$ for the $i$-th symbol and $n$-th subcarrier is given by $\mathbf{x}_{b,i,n} = \sum_{k \in \mathcal{K}} \mathbf{w}_{b,k,n}(a_t) s_{k,i,n}$, where $\mathbf{w}_{b,k,n}(a_t)$ is the precoding vector corresponding to the configuration $a_t$, and $s_{k,i,n}$ is the data symbol. The received signal at user $k$ for the $n$-th subcarrier is expressed as $y_{k,i,n} = \sum_{b \in \mathcal{B}} \mathbf{h}_{b,k,n}^H(a_t) \mathbf{x}_{b,i,n} + z_{k,i,n}$, where $z_{k,i,n}$ denotes the additive white Gaussian noise (AWGN). The instantaneous signal-to-interference-plus-noise ratio (SINR) for user $k$ on subcarrier $n$ is formulated as:
\begin{equation}
\gamma_{k,n}(a_t) = \frac{\left| \sum_{b \in \mathcal{B}} \mathbf{h}_{b,k,n}^H(a_t) \mathbf{w}_{b,k,n}(a_t) \right|^2}{\sum_{k' \neq k} \left| \sum_{b \in \mathcal{B}} \mathbf{h}_{b,k,n}^H(a_t) \mathbf{w}_{b,k',n}(a_t) \right|^2 + \sigma_k^2}.
\end{equation}
Accordingly, the steady-state sum-rate for the dwell stage is defined as 
$R_{\mathrm{d}}(a_t) = \frac{1}{|\mathcal{N}|} \sum_{n \in \mathcal{N}} \sum_{k \in \mathcal{K}} \log_2(1 + \gamma_{k,n}(a_t)).$

We consider a fully cooperative sensing mode, where all BSs in $\mathcal{B}$ simultaneously act as sensing transmitters and receivers. For a given FAS configuration $a_t$, the sensing observation collected at each BS $r \in \mathcal{B}$ is expressed as:
\begin{equation}
\mathbf{y}_{r} = \mathbf{A}_{r}(a_t )\mathbf{g}_{\mathrm{s}} 
+ \mathbf{H}_{\mathrm{SI},r}(a_t)\mathbf{x}_r 
+ \sum_{b \in \mathcal{B} \setminus \{r\}} \mathbf{H}_{\mathrm{CI},r,b}(a_t)\mathbf{x}_b 
+ \mathbf{n}_{r},
\label{eq:rx_sensing_model}
\end{equation}
where $\mathbf{g}_{\mathrm{s}} \in \mathbb{C}^{L \times 1}$ denotes the target response vector for the $L$ targets, and $\mathbf{A}_{r}(a_t )$ is the sensing dictionary at receiver $r$, which incorporates the aggregate echoes from all transmitting BSs. The terms $\mathbf{H}_{\mathrm{SI},r}(a_t)$ and $\mathbf{H}_{\mathrm{CI},r,b}(a_t)$ represent the residual self-interference (SI) channel at BS $r$ and the cross-interference (CI) channel from BS $b$ to BS $r$, respectively. By stacking the observations from all BSs, the joint sensing model is given by:
\begin{equation}
\mathbf{y}_{\mathrm{s}} = \mathbf{A}_{\mathrm{s}}(a_t )\mathbf{g}_{\mathrm{s}} + \mathbf{z}_{\mathrm{s}}(a_t) + \mathbf{n}_{\mathrm{s}},
\label{eq:joint_sensing_model_si}
\end{equation}
where $\mathbf{A}_{\mathrm{s}}(a_t ) = [\mathbf{A}_{1}^{T}, \dots, \mathbf{A}_{B}^{T}]^{T}$ is the global sensing dictionary, $\mathbf{z}_{\mathrm{s}}(a_t)$ denotes the aggregate residual SI and CI, and $\mathbf n_{\mathrm s}$ denotes the stacked sensing noise. After interference suppression, the effective disturbance is modeled by the covariance matrix $\mathbf{R}_{\mathrm{e}}(a_t) = \sigma_{\mathrm{s}}^2\mathbf{I} + \mathbf{R}_{\mathrm{inf}}(a_t)$. Assuming $\mathbf{g}_{\mathrm{s}} \sim \mathcal{CN}(\mathbf{0}, \mathbf{R}_{\mathrm{g}})$, the configuration-level sensing mutual information (MI) is defined as:
\begin{equation}
\mathcal{I}_{\mathrm{s}}(a_t) = \log_2\det \left( \mathbf{I} + \mathbf{R}_{\mathrm{e}}^{-1}(a_t) \mathbf{A}_{\mathrm{s}}(a_t ) \mathbf{R}_{\mathrm{g}} \mathbf{A}_{\mathrm{s}}^{H}(a_t ) \right).
\label{eq:si_aware_multistatic_mi}
\end{equation}

\section{Cost-aware FAS Switching via OCBO Method}
This section develops the cost-aware FAS-ISAC switching framework. First, we derive cost-aware performance metrics that characterize the slot-level communication and sensing objectives under both port movement and stable dwell stages. Then, we formulate the resulting stay-or-switch decision as a multi-objective online optimization problem.  Finally, we propose the OCBO algorithm to enable robust FAS reconfiguration.

\begin{figure}[t]
    \centering
    \includegraphics[width=0.90\linewidth]{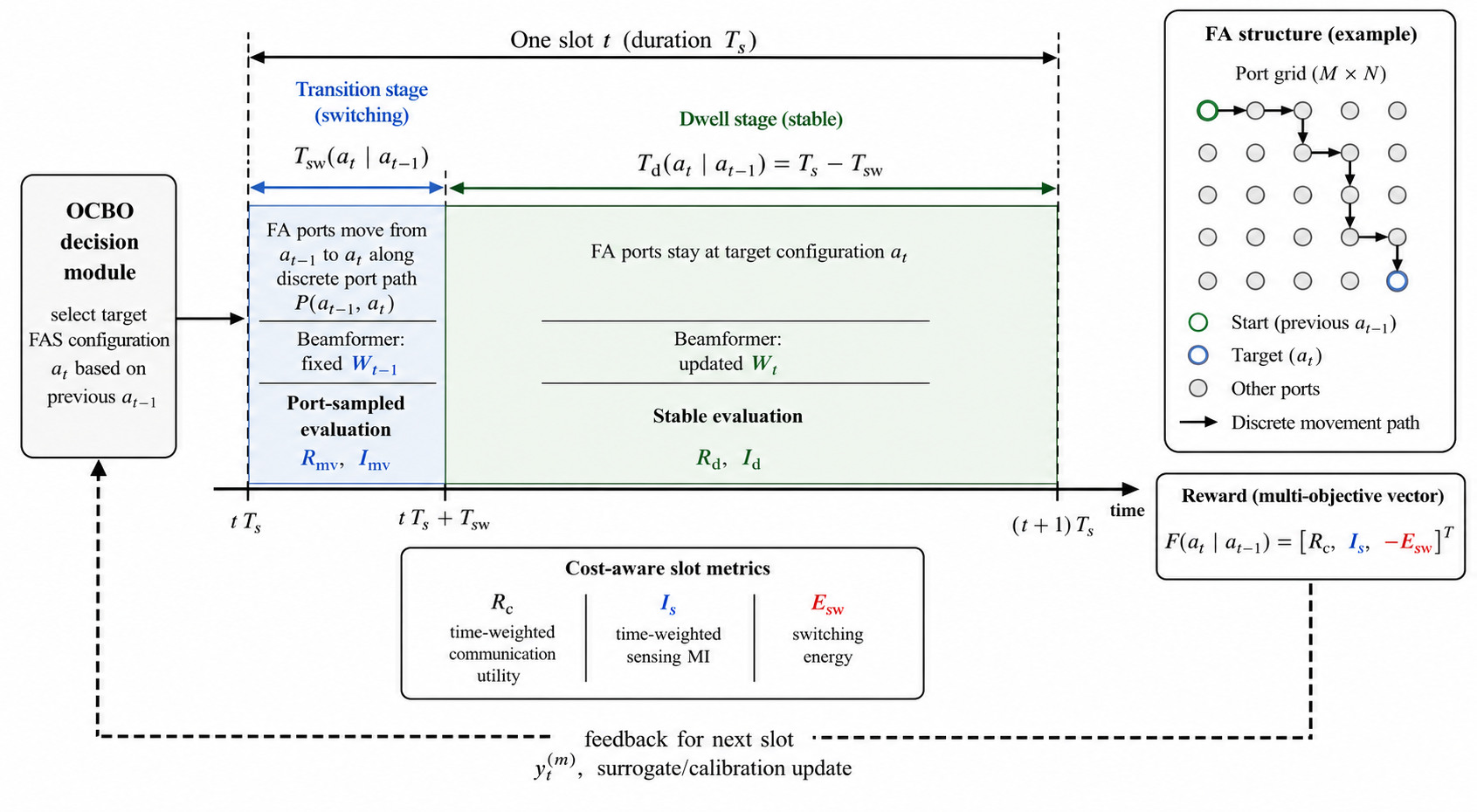}
    \caption{Cost-aware slot structure for FAS-ISAC port switching}
    \label{fig:slot_structure}
\end{figure}

\subsection{Cost-Aware Performance Metrics}
The movement of FAS ports introduces a time-varying system geometry during the switching interval, making the slot-level performance dependent on the switching trajectory between the preceding configuration $a_{t-1}$ and the target configuration $a_t$. Here, the slot-level communication utility is defined as the weighted average of the rates achieved in each stage:
\begin{equation}
R_{\mathrm{c}}(a_t|a_{t-1}) = \frac{T_{\mathrm{sw}}}{T_{\mathrm{s}}} R_{\mathrm{mv}}(a_t|a_{t-1}) + \frac{T_{\mathrm{d}}}{T_{\mathrm{s}}} R_{\mathrm{d}}(a_t, \mathbf{W}(a_t)),
\label{eq:slot_comm_performance}
\end{equation}
where $R_{\mathrm{d}}(a_t, \mathbf{W}(a_t))$ is the steady-state rate at the target configuration. $T_{\mathrm{d}} = T_{\mathrm{s}} - T_{\mathrm{sw}}(a_t | a_{t-1})$ represents the subsequent dwell stage. $R_{\mathrm{mv}}$ is characterized as follows.

\begin{proposition}\label{prop:R_mv_approx}
In the high-port-density regime, the switching-stage rate can be approximated by the average of the instantaneous rates evaluated at the available ports along the movement path from $a_{t-1}$ to $a_t$, 
\begin{equation}
R_{\mathrm{mv}}(a_t|a_{t-1})
\approx
\frac{1}{|\mathcal P|}
\sum_{a_p\in\mathcal P}
R_{\mathrm c}^{\mathrm{ins}}(a_p,\mathbf W_{t-1}),
\label{eq:Rmv_port_sampled}
\end{equation}
where $\mathcal P$ denotes the set of available ports along the movement path, and $R_{\mathrm c}^{\mathrm{ins}}(a_n,\mathbf W_{t-1})$ denotes the instantaneous sum-rate evaluated at port configuration $a_n$ using $\mathbf W_{t-1}$, i.e.,
$R_{\mathrm c}^{\mathrm{ins}}(a_n,\mathbf W_{t-1})=\frac{1}{|\mathcal N|}\sum_{n\in\mathcal N}\sum_{k\in\mathcal K}\log_2(1+\gamma_{k,n}(a_n,\mathbf W_{t-1}))$.
\end{proposition}
\begin{proof}
Let $a(\tau)$ denote the continuous FAS configuration during the switching interval $\tau\in[0,T_{\mathrm{sw}}]$. During the switching, the beamformer is fixed as $\mathbf W_{t-1}=\mathbf W(a_{t-1})$, and the switching-stage rate is
\begin{equation}
R_{\mathrm{mv}}(a_t|a_{t-1})
=
\frac{1}{T_{\mathrm{sw}}}
\int_0^{T_{\mathrm{sw}}}
R_{\mathrm c}^{\mathrm{ins}}(a(\tau),\mathbf W_{t-1})d\tau .
\label{eq:Rmv_continuous}
\end{equation}

For FAS channels, the spatial correlation between two nearby positions separated by distance $d$ can be characterized by $\rho_h(d)=J_0(2\pi d/\lambda)$, where $J_0(\cdot)$ is the zero-order Bessel function. In the high-port-density regime, $d/\lambda$ is small, and using $J_0(x)=1-x^2/4+O(x^4)$ gives $\rho_h(d)=1-\pi^2 d^2/\lambda^2+O(d^4/\lambda^4)$. $\rho_h(d)$ approaches one as $d/\lambda \rightarrow 0$, indicating strong channel correlation within each local port neighborhood. Therefore, the continuous rate around each local segment can be represented by the rate evaluated at the corresponding port. Accordingly, the integral in \eqref{eq:Rmv_continuous} can be approximated by a weighted sum over the ports along the movement path,
\begin{equation}
R_{\mathrm{mv}}(a_t|a_{t-1})
\approx
\sum_{a_n\in\mathcal P(a_{t-1},a_t)}
\omega_n R_{\mathrm c}^{\mathrm{ins}}(a_n,\mathbf W_{t-1}).
\end{equation}
With uniformly spaced ports and constant movement speed, the local weights are approximately identical, i.e., $\omega_n\approx 1/|\mathcal P(a_{t-1},a_t)|$. Thus,
\begin{equation}
R_{\mathrm{mv}}(a_t|a_{t-1})
\approx
\frac{1}{|\mathcal P(a_{t-1},a_t)|}
\sum_{a_n\in\mathcal P}
R_{\mathrm c}^{\mathrm{ins}}(a_n,\mathbf W_{t-1}).
\end{equation}
By writing $R_{\mathrm c}^{\mathrm{ins}}(a_n)=R_{\mathrm c}^{\mathrm{ins}}(a_n,\mathbf W_{t-1})$ during the switching-stage, \eqref{eq:Rmv_port_sampled} follows.
\end{proof}

\begin{lemma}\label{lem:I_mv_approx}
In the high-port-density regime, the switching-stage sensing MI can be approximated by the average of the configuration-level sensing MI evaluated at the available ports along the movement path:
\begin{equation}
I_{\mathrm{mv}}(a_t|a_{t-1})
\approx
\frac{1}{|\mathcal P(a_{t-1},a_t)|}
\sum_{a_n\in\mathcal P(a_{t-1},a_t)}
\mathcal I_{\mathrm{s}}(a_n,\mathbf W_{t-1}),
\label{eq:Imv_port_sampled}
\end{equation}
where $\mathcal P(a_{t-1},a_t)$ denotes the set of available ports along the movement path, and $\mathcal I_{\mathrm{s}}(a_n,\mathbf W_{t-1})$ denotes the instantaneous sensing MI evaluated at port configuration $a_n$ using $\mathbf W_{t-1}$.
\end{lemma}

Similarly, the slot-level sensing utility is modeled as
\begin{equation}
I_{\mathrm{s}}(a_t|a_{t-1}) =
\frac{T_{\mathrm{sw}}}{T_{\mathrm{s}}} I_{\mathrm{mv}}(a_t|a_{t-1})
+
\frac{T_{\mathrm{d}}}{T_{\mathrm{s}}} I_{\mathrm{d}}(a_t,\mathbf W(a_t)),
\label{eq:slot_sensing_performance}
\end{equation}
where $I_{\mathrm{d}}(a_t,\mathbf W(a_t))=\mathcal I_{\mathrm{s}}(a_t,\mathbf W(a_t))$.
According to Lemma~\ref{lem:I_mv_approx}, the switching-stage sensing MI can be evaluated by the port-sampled approximation in \eqref{eq:Imv_port_sampled}.

The switching energy is computed according to the total movement time of all activated FA elements:
\begin{equation}
E_{\mathrm{sw}}(a_t|a_{t-1})
=
\sum_{b,m}
P_{\mathrm{mv}}
\frac{d_{\mathrm{path}}^{(b,m)}(a_{t-1},a_t)}{v_{\mathrm{FA}}},
\label{eq:switching_energy}
\end{equation}
where $P_{\mathrm{mv}}$ denotes the movement power.

\subsection{Problem Formulation}
Based on the above discussion, the reward vector capturing the multi-objective trade-off is defined as:
\begin{equation}
\mathbf{F}(a_t|a_{t-1}) =
\left[
R_{\mathrm{c}}(a_t|a_{t-1}),
I_{\mathrm{s}}(a_t|a_{t-1}),
-E_{\mathrm{sw}}(a_t|a_{t-1})
\right]^T,
\label{eq:slot_reward_vector}
\end{equation}
where the negative sign indicates that the switching energy is minimized while communication and sensing objectives are maximized. Then, the optimization problem is formulated as:
\begin{subequations} \label{eq:opt_problem}
\begin{align}
 \text{(P1)} \quad &
\max_{\{a_t\}_{t=1}^T} \quad 
\sum_{t=1}^T \mathbf{F}(a_t | a_{t-1}) 
\notag \\ 
\text{s.t.} \quad 
& \mathcal{M}_{b,t}^{\mathrm{tx}} \subseteq \mathcal{C}_b, \quad 
  \mathcal{M}_{b,t}^{\mathrm{rx}} \subseteq \mathcal{C}_b, 
  \quad \forall b \in \mathcal{B}, 
  \label{eq:port_selection} \\
& |\mathcal{M}_{b,t}^{\mathrm{tx}}| = M_b, \quad 
  |\mathcal{M}_{b,t}^{\mathrm{rx}}| = N_b, 
  \quad \forall b \in \mathcal{B}, 
  \label{eq:port_cardinality} \\
& \mathcal{M}_{b,t}^{\mathrm{tx}} \cap 
  \mathcal{M}_{b,t}^{\mathrm{rx}} = \emptyset, 
  \quad \forall b \in \mathcal{B}, 
  \label{eq:exclusivity} \\
& T_{\mathrm{sw}}(a_t | a_{t-1}) \leq T_{\mathrm{s}}, 
  \quad \forall t \in \mathcal{T}, 
  \label{eq:switching_time} \\
& \mathbf{W}(a_t) = \mathcal{B}(\widehat{\mathbf{H}}(a_t)), 
  \quad \forall t \in \mathcal{T}, 
  \label{eq:bf_policy} \\
& \operatorname{Tr} \left( 
  \mathbf{W}_{b,n}(a_t) \mathbf{W}_{b,n}^H(a_t) 
  \right) \leq P_{b,\max}, 
  \quad \forall b, n, 
  \label{eq:power_budget}
\end{align}
\end{subequations}
where \eqref{eq:port_selection}--\eqref{eq:exclusivity} define the feasible port selection and exclusivity constraints, \eqref{eq:switching_time} ensures the reconfiguration completes within one slot, and \eqref{eq:bf_policy}--\eqref{eq:power_budget} specify the precoding policy and power constraints. Since $\mathbf{F}(a_t|a_{t-1})$ is a vector-valued reward, the maximization in \text{(P1)} is pursued in the Pareto sense. Here, $\mathcal{C}_b$ is the candidate port set on its FAS. The constants $M_b$ and $N_b$ represent the numbers of activated transmit and receive ports at BS $b$, respectively, and $\mathcal{T}$ denotes the set of time slots. The precoding matrix $\mathbf W(a_t)$ is obtained from the precoding rule $\mathcal{B}(\hat{\mathbf H}(a_t))$ based on the pilot feedback, and $P_{b,\max}$ denotes the per-subcarrier transmit power budget of BS $b$. 

The ISAC objectives depend on unknown instantaneous channels, target responses, and residual interference, which can be represented as implicit mappings $R_{\mathrm c}(a_t|a_{t-1})=\Phi_{\mathrm c}(a_{t-1},a_t,\mathbf H_t)$ and $I_{\mathrm s}(a_t|a_{t-1})=\Phi_{\mathrm s}(a_{t-1},a_t,\mathbf A_{\mathrm{s},t},\mathbf R_{\mathrm{SI},t})$. Therefore, part of the objectives in \text{(P1)} have a \textit{gray-box} structure, which cannot be calculated in closed form by the white box methods before the configuration is evaluated.

\subsection{Cost-aware GP Surrogates}
The whole slot structure is shown in Fig.~\ref{fig:slot_structure}. In the proposed FAS switching process, each action corresponds to a port switching from the previous FAS configuration $a_{t-1}$ to the target configuration $a_t$. After executing this switching, the system observes noisy feedback from the realized ISAC objectives:
\begin{equation}
y_t^{(m)} = F_m(a_t|a_{t-1}) + \xi_t^{(m)}, \quad m \in \{1,2\},
\label{eq:bandit_feedback}
\end{equation}
where $m=1$ and $m=2$ correspond to the observations of $R_{\mathrm c}$ and $I_{\mathrm s}$, respectively, and $\xi_t^{(m)}$ denotes measurement noise. The feedback is then used by the proposed OCBO framework to update the surrogate models and guide the online search. Specifically, two independent Gaussian Process (GP) surrogates are deployed to learn the two cost-aware gray-box physical objectives ($m \in \{1,2\}$). Since the slot-level feedback depends on the FAS port switching from $a_{t-1}$ to $a$, the surrogate input is defined over the switching pair $(a_{t-1}, a)$ rather than the target port alone. The surrogate can be modeled as:
\begin{equation}
F_m(a | a_{t-1}) \sim \mathcal{GP}\left(\mu_m^{\mathrm{g}}(a | a_{t-1}, \mathbf{b}_t), k_m((a, a_{t-1}), \cdot)\right).
\end{equation}

Given the historical evaluation dataset $\mathcal{D}_t^{(m)} = \{((a_\tau, a_{\tau-1}, s_\tau), y_\tau^{(m)})\}_{\tau=1}^t$, the GP predictive distribution for a candidate switching is \cite{Daulton2022MORBO}:
\begin{equation}
\label{eq:surrogate_posterior}
p(y^{(m)} | a, a_{t-1}, \mathcal{D}_t^{(m)}) = \mathcal{N}\left( \mu_m, \sigma_m^2 \right).
\end{equation}

In the proposed cost-aware FAS switching architecture, the feedback associated with a target configuration $a_t$ is not solely determined by its dwell performance, but also by the preceding configuration $a_{t-1}$. Consequently, the observed utility may exhibit action-dependent fluctuations and temporal drift that are not fully captured by the standard GP posterior variance. This can make the GP surrogate overconfident in dynamically changing regions, thereby misleading the acquisition function. To mitigate this issue, we introduce conformal calibration to adjust the predictive uncertainty using residual feedback, enabling more reliable stay-or-switch decisions.

\begin{algorithm}[t]
\caption{The OCBO-Based FAS Port Switching Algorithm}
\label{alg:ocbo_fas}
\begin{algorithmic}[1]
\STATE \textbf{Input:} Feasible FAS configuration space $\mathcal A$; initial FAS configuration $a_0$; miscoverage level $\alpha$; energy penalty weight $\eta_{\rm e}$; time horizon $T$.
\STATE \textbf{Initialization:} Fit cost-aware GP surrogates using initial evaluations $\mathcal{D}_0$.
\FOR{each time slot $t = 1, 2, \dots, T$}
    \STATE Observe the previous configuration $a_{t-1}$ and construct the feasible candidate set $\mathcal{C}_t\subseteq\mathcal{A}$ according to the constraints in \eqref{eq:port_selection}--\eqref{eq:switching_time}.
    \STATE Compute the OCBO acquisition score $Q_t(a)$ for each feasible candidate $a\in\mathcal{C}_t$ according to \eqref{eq:ocbo_acquisition}.
    \STATE Select target configuration by maximizing \eqref{eq:ocbo_acquisition}.
    \STATE Deploy $a_t$ and observe $\mathbf{y}_t=[y_t^{(1)},y_t^{(2)}]^T$ according to \eqref{eq:bandit_feedback}.
    \FOR{$m \in \{1, 2\}$}
        \STATE Update dataset $\mathcal{D}_t^{(m)} \leftarrow \mathcal{D}_{t-1}^{(m)} \cup \{(\phi_t(a_t), y_t^{(m)})\}$.
        \STATE Fit the GP surrogate using $\mathcal{D}_t^{(m)}$ and obtain the posterior in \eqref{eq:surrogate_posterior}.
        \STATE Update $\mathcal{S}_{t,m}$ using residual $s_{t,m}$ in \eqref{eq:residual_score}.
        \STATE Update $q_{t,m}$ according to \eqref{eq:calibration_factor}.
    \ENDFOR
\ENDFOR
\STATE \textbf{Output:} Deployed configuration sequence $\{a_t\}_{t=1}^{T}$ and collected observations.
\end{algorithmic}
\end{algorithm}


\subsection{Online Conformal Calibration}
To robustify the uncertainty estimates against model misspecification caused by FAS port switching, beamformer mismatch, and environmental drift, we apply online conformal calibration \cite{Stanton2023ConformalBO}. After observing the noisy feedback $y_t^{(m)}$ and computing the GP prediction mean $\mu_m(a_t|a_{t-1})$ and standard deviation $\sigma_m(a_t|a_{t-1})$, we define the residual score:
\begin{equation}
\label{eq:residual_score}
s_{t,m} = \frac{|y_t^{(m)} - \mu_m(a_t|a_{t-1})|}{\sigma_m(a_t|a_{t-1}) + \epsilon},
\end{equation}
where $\epsilon$ is a small positive constant. We maintain a recent calibration buffer $\mathcal{S}_{t,m} = \{s_{t-W,m}, \dots, s_{t-1,m}\}$. Based on this online buffer, we compute the conformal calibration factor as the empirical quantile for a target miscoverage risk $\alpha$ as: 
\begin{equation}
\label{eq:calibration_factor}
q_{t,m} = \mathrm{Quantile}_{1-\alpha}(\mathcal{S}_{t,m}),
\end{equation}
where $\mathrm{Quantile}_{1-\alpha}(\mathcal{S}_{t,m})$ returns the value below which approximately a fraction $1-\alpha$ of the residual scores in $\mathcal{S}_{t,m}$ fall. The GP predictive uncertainty is then robustly scaled as
$\tilde{\sigma}_{t,m}(a)=q_{t,m}\sigma_m(a|a_{t-1})$.

We adopt the expected hypervolume improvement (EHI) as the acquisition function to guide the multi-objective FAS switching search \cite{Daulton2021qNEHVI}. Let $\mathcal{P}_t$ denote the current approximated Pareto set and $\mathbf r$ be the reference point. For a predicted objective sample $\tilde{\mathbf F}(a|a_{t-1})$, the hypervolume improvement is $\mathrm{HVI}_t(a)=\mathrm{HV}\!\left(\mathcal{P}_t\cup\{\tilde{\mathbf F}(a|a_{t-1})\};\mathbf r\right)-\mathrm{HV}\!\left(\mathcal{P}_t;\mathbf r\right)$, where $\mathrm{HV}(\cdot;\mathbf r)$ denotes the hypervolume dominated by a Pareto set and bounded by $\mathbf r$. The EHI score is then given by $
\mathrm{EHI}_t(a|a_{t-1})
=
\mathbb{E}_{\tilde{\mathbf F}}
\left[
\mathrm{HVI}_t(a)
\right]$, where $\tilde{\mathbf F}$ is sampled from the conformal-calibrated predictive distribution. Based on the calibrated EHI score, the deterministic FAS port movement energy is further introduced as an explicit penalty, yielding the following cost-aware acquisition function:
\begin{equation}
\label{eq:ocbo_acquisition}
\mathcal{Q}_{t}(a)
=
\mathrm{EHI}_t(a|a_{t-1})
-
\eta_{\mathrm e}E_{\mathrm{mv}}(a|a_{t-1}),
\end{equation}
where $\eta_{\mathrm e}$ controls the trade-off between the expected physical performance improvement and the switching energy cost. The next FAS configuration is selected by maximizing $\mathcal Q_t(a)$. For the stay action ($a = a_{t-1}$), its cost-aware score is simply $\mathrm{EHI}_t(a_{t-1}|a_{t-1})$. By combining the known analytical energy cost with the conformal-calibrated physical uncertainty, the framework actively avoids risky or overly frequent FAS port switchings. The procedure is summarized in Algorithm \ref{alg:ocbo_fas}.

\begin{figure}[t]
    \centering
    \includegraphics[width=0.75\linewidth]{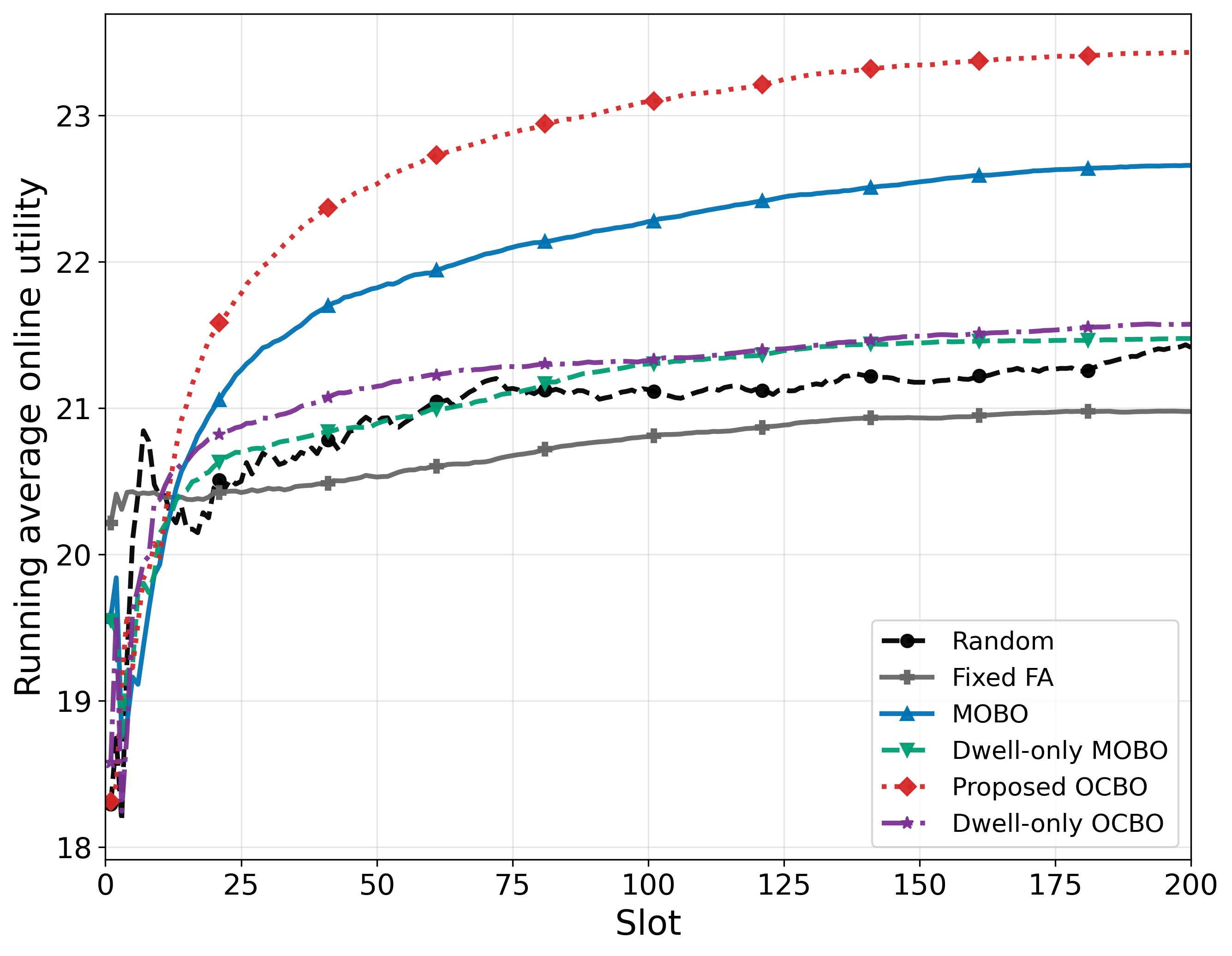}
    \caption{Running average online cost-aware utility over $200$ time slots.}
    \label{fig:running_utility}
\end{figure}

\begin{figure}[t]
    \centering
    \includegraphics[width=0.75\linewidth]{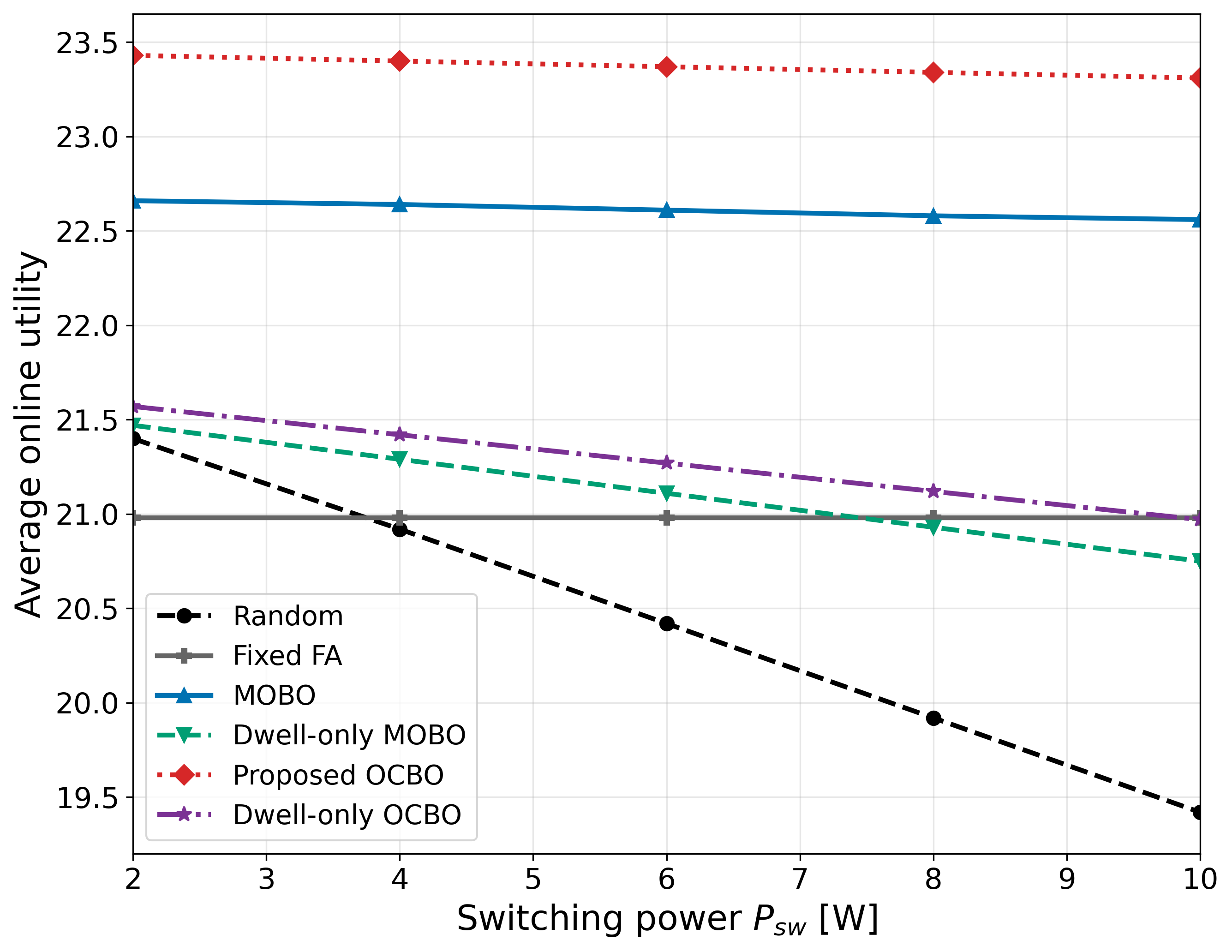}
    \caption{Average online cost-aware utility versus switching power.}
    \label{fig:power_sweep}
\end{figure}

\begin{figure}[t]
    \centering
    \includegraphics[width=0.75\linewidth]{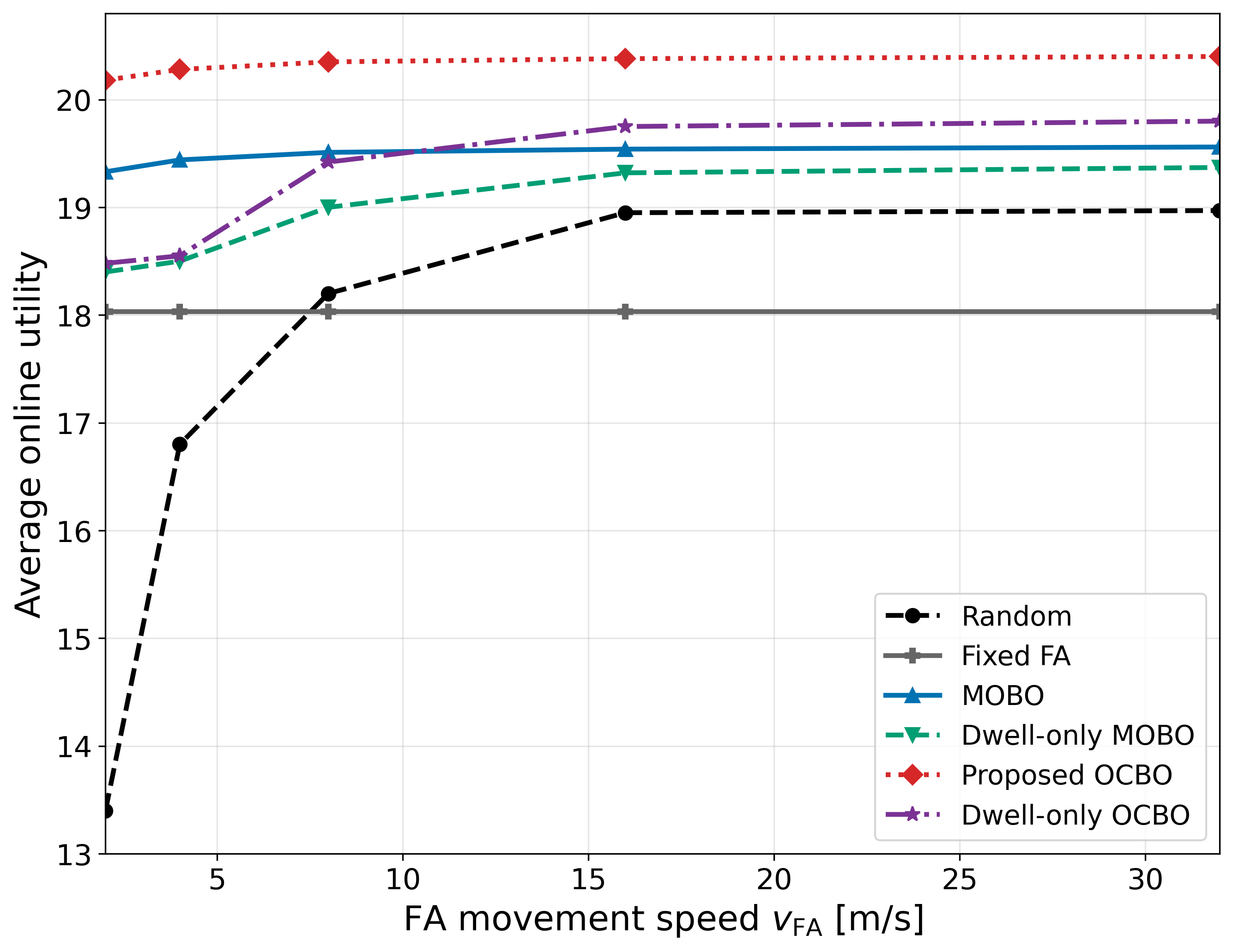}
    \caption{Average online cost-aware utility versus FA movement speed.}
    \label{fig:speed_sweep}
\end{figure}

\section{Simulations}
In this section, we evaluate the performance of the proposed cost-aware FAS-ISAC switching algorithm through numerical simulations. The simulation scenario is considered within a 3D space of $200 \times 200 \times 80$~m$^3$, where $B=3$ fixed BSs are deployed. The system serves $K=3$ dynamic communication users moving at speeds up to $10$~m/s, while simultaneously monitoring $L=2$ dynamic point targets (UAVs) via a fully cooperative networked sensing architecture. These UAV targets move at velocities of approximately $15$~m/s, necessitating real-time configuration adaptation. Each BS is equipped with a fluid antenna surface featuring a $16 \times 16$ candidate port grid (256 ports in total) with a port spacing of $0.25\lambda$. In each time slot $t$, each BS activates $M_b=2$ ports for communication transmission and $N_b=2$ ports for sensing reception. The system operates at a carrier frequency of $28$~GHz using an OFDM waveform with a $128$-point FFT and $|\mathcal{N}|=32$ active subcarriers. The transmit power is set to $40$~dBm. The channel model incorporates Rician fading with a $K$-factor of $10$~dB and a pathloss model with an exponent of $2.5$. The online decision-making process spans $T=200$ slots with a slot interval $T_{\mathrm{s}}=0.02$~s, totaling a duration of $4$~s.

We compare the proposed cost-aware FAS switching architecture with four baselines. \textit{Random} selects the FAS configuration uniformly from the candidate set. \textit{Fixed} keeps the initial port configuration unchanged during all slots. \textit{MOBO} adopts standard multi-objective Bayesian optimization without conformal calibration. \textit{Dwell-only} uses multi-objective search strategy but evaluates only the dwell-stage objectives, ignoring the switching-stage contribution. These baselines are used to verify the benefits of cost-aware performance modeling and conformal uncertainty calibration.

Fig. \ref{fig:running_utility} illustrates the running average online cost-aware utility across $200$ slots. It can be observed that the proposed OCBO method significantly outperforms all other baselines, reaching a superior long-term utility of approximately $23.5$. While MOBO initially tracks closely with the proposed method, its performance gap widens as the environment evolves, primarily because standard GP surrogates struggle with miscalibration under the dynamic mismatches induced by user mobility. In contrast, the conformal calibration in OCBO robustly adjusts uncertainty estimates, preventing the optimizer from making overly aggressive or suboptimal switching decisions. The \textit{Dwell-only} baseline, which ignores switching-stage contributions, exhibit a noticeable performance degradation. This confirms that modeling the unstable yet usable signals during antenna movement is crucial for maximizing long-term reward. 

Fig.~\ref{fig:power_sweep} shows the average online cost-aware utility under different switching power levels. As $P_{\mathrm{sw}}$ increases, the movement energy penalty becomes more significant, leading to a gradual utility reduction for switching-based schemes. The proposed OCBO consistently achieves the highest utility across all tested power levels, indicating that the cost-aware acquisition function effectively suppresses low-benefit switching when the mechanical cost becomes large. In contrast, \textit{Random} suffers a sharp degradation because it frequently triggers unnecessary port movements without considering the energy overhead. The \textit{Dwell-only} variants also exhibit lower performance since their decisions are made without explicitly accounting for the usable switching-stage performance, resulting in a suboptimal performance-cost tradeoff.

Fig.~\ref{fig:speed_sweep} evaluates the impact of the FA movement speed $v_{\mathrm{FA}}$. A higher movement speed shortens the switching duration and reduces the switching overhead, thereby improving the average utility of most adaptive schemes. The proposed OCBO maintains the best performance over the entire speed range, showing its robustness to different hardware movement capabilities. Compared with dwell-only MOBO and dwell-only OCBO, the gain of the proposed method confirms the benefit of incorporating switching-stage communication and sensing objectives into the slot-level evaluation. The \textit{Fixed} baseline remains almost unchanged because no port movement is performed, while the \textit{Random} baseline improves with speed but remains inferior due to its lack of informed switching decisions.

\section{Conclusions}
This paper investigated a smart switching mechanism for FAS to achieve improved ISAC performance. A cost-aware framework was developed that incorporates both switching-stage and dwell-stage ISAC objectives. The online conformal Bayesian optimization was adopted to learn the unknown physical performance while calibrating surrogate uncertainty under environmental drift. An acquisition function was further tailored to balance the calibrated utility improvement against deterministic movement energy. Simulation results demonstrate that the proposed switching scheme achieves a better long-term performance–cost trade-off than conventional cost-agnostic schemes, by exploiting usable switching-stage objectives and reducing unnecessary port movements. Future work will address complex dynamic environments with abrupt distribution shifts, along with extended-target sensing and tighter communication–sensing integration for more consistent ISAC performance.

\bibliographystyle{IEEEtran}
\bibliography{refs_conference}

\end{document}